\documentclass[review,hidelinks,onefignum,onetabnum]{siamart251216} 
 
\usepackage{lipsum}
\usepackage{amsfonts}
\usepackage{graphicx}
\usepackage{epstopdf}
\usepackage{algorithmic}
\usepackage{bm}

\ifpdf
  \DeclareGraphicsExtensions{.eps,.pdf,.png,.jpg}
\else
  \DeclareGraphicsExtensions{.eps}
\fi

\newsiamremark{remark}{Remark}
\newsiamremark{hypothesis}{Hypothesis}
\crefname{hypothesis}{Hypothesis}{Hypotheses}
\newsiamthm{claim}{Claim}
\newsiamremark{fact}{Fact}
\crefname{fact}{Fact}{Facts}

\headers{SPH Model for Droplet Dynamics}{Z. H. Qiao,  Y. F. Wei and  X. M. Xu}

\title{An SPH Model with Physically Prescribed Parameters for Droplet Dynamics on Complex Surfaces \thanks{Submitted to the editors DATE.
\funding{This work is supported by the CAS AMSS-PolyU Joint Laboratory of Applied Mathematics (No. JLFS/P-501/24).
The first author was partially supported by the Hong Kong Research Grants Council GRF grant 15305624 and NSFC/RGC Joint Research Scheme (No. N\_PolyU5145/24).
The second author was partially supported by the Hong Kong Polytechnic University Postdoctoral Research Fund 1-W30N.
}}} 

\author{ZHONGHUA QIAO \thanks{Department of Applied Mathematics, The Hong Kong Polytechnic University,  Hung Hom, Kowloon, Hong Kong
  (\email{zhonghua.qiao@polyu.edu.hk}).}
\and YIFAN WEI \thanks{Corresponding author. Department of Applied Mathematics, The Hong Kong Polytechnic University,  Hung Hom, Kowloon, Hong Kong
  (\email{yi-fan.wei@polyu.edu.hk}).}
\and XIANMIN XU \thanks{NCMIS \& LSEC,Institute of Computational Mathematics and Scientific/Engineering Computing, Academy of Mathematics and Systems Science, Beijing 100190, PR China
  (\email{xmxu@lsec.cc.ac.cn}).}}

\usepackage{amsopn}

\usepackage{xcolor}
\usepackage{booktabs, comment}
\usepackage{pgfpages}
\usepackage{csquotes}
\usepackage{amsmath}
\usepackage{pgfplots}
\usepackage{tikz}
\usepackage{tikz-3dplot}
\usepackage{graphicx} 
\usepackage{float}
\usepackage{subfigure}
 
\usepackage{multirow}
\usepackage{epstopdf}
\usepackage{hyperref}
\usetikzlibrary{angles,quotes}

\begin{document}
\nolinenumbers

\maketitle

\begin{abstract}
Numerical simulation of droplet dynamics on complex surfaces with varying wettability is of great significance to both engineering applications and fundamental research. However, existing numerical methods still face challenges in accurately capturing interfacial interactions while preserving physical consistency and computational efficiency. In this work, a physically grounded and efficient smoothed particle hydrodynamics (SPH) model is developed for droplet dynamics simulation. To reduce computational cost, a single-phase droplet modeling strategy is employed. At the interface, long-range interactions are approximated using the SPH kernel function, whereas short-range interactions are represented through pressure. Based on this treatment, an explicit relationship between the intermolecular potential energy and the macroscopic surface tension coefficient is further established, thereby reducing reliance on empirical parameter calibration. The proposed method is first validated through static wetting simulations, where the equilibrium contact angles agree well with the Young--Dupr\'e equation. Further simulations of wetting and droplet impact demonstrate that the method is capable of capturing complex dynamic wetting behaviors.
\end{abstract} 



\section{Introduction}
Droplet dynamics has attracted sustained attention due to its broad relevance to both engineering applications and fundamental science. Representative examples include inkjet printing \cite{wijshoff2018drop,lohse2022fundamental}, aircraft anti-icing and de-icing \cite{long2023micro,yang2023icephobic,tian2013wetting}, agricultural spraying \cite{privitera2023drop,wang2023sustained}, and microfluidic systems \cite{moragues2023droplet,kovalchuk2023review}. 
In these processes, droplet behavior is governed by the interplay of inertia, viscosity, surface tension, and solid--liquid interactions, leading to complex dynamic phenomena such as spreading, receding, rebounding, and splashing \cite{mohammad2023physics}. 
Nevertheless, understanding and accurately predicting interfacial interaction forces, particularly those associated with surface tension and wettability, remain major challenges in the study of droplet dynamics.

Numerical methods have significantly advanced the simulation and understanding of droplet dynamics. Among them, traditional mesh-based approaches have played a central role in multiphase flow modeling and have yielded many important results. For instance, a geometric Volume-of-Fluid-based framework was proposed for multicomponent phase change and validated for non-isothermal two-phase flows through a series of benchmark tests \cite{cipriano2024multicomponent}. Similarly, a thermodynamically consistent lattice Boltzmann method was developed for non-isothermal two-phase flows with liquid--vapour phase change and wetting effects \cite{qiao2026unified}. 
High-order mesh-based formulations have also been developed for compressible multiphase flows; for instance, Qin et al.~\cite{qin2025high} proposed a TENO scheme within the lattice Boltzmann flux solver framework, combined with a level-set-based interface-tracking method to resolve phase interfaces.
However, droplet dynamics often involves large interfacial deformation, moving contact lines, topological changes, and repeated contact or collision events \cite{jiao2023effect}, which pose significant challenges to mesh-based methods in terms of interface accuracy and the treatment of complex boundary evolution. In many cases, additional interface-capturing or interface-tracking techniques are required, further increasing the algorithmic complexity. In contrast, smoothed particle hydrodynamics (SPH), as a fully Lagrangian meshfree method, has shown great potential for droplet dynamics simulation \cite{monaghan1992smoothed}. Owing to its particle-based discretization, SPH can naturally handle large deformation, free-surface motion, and complex interfacial evolution without suffering from mesh entanglement, making it a particularly attractive approach for droplet impact and other interfacial flow problems.

Considerable efforts have been devoted to modeling interfacial interaction forces in SPH, as these forces are essential for accurately representing surface tension and  wettability. Existing SPH treatments can generally be classified into two main categories. The first is the continuum surface force (CSF) approach, originally proposed by Brackbill et al. \cite{brackbill1992continuum}, in which surface tension is transformed into an equivalent volumetric force, thereby enabling interfacial effects to be incorporated within a continuum framework. In fact, this approach has been widely adopted in SPH simulations of multiphase and free-surface flows. For example, Vergnaud et al. \cite{vergnaud2022c} introduced several important improvements for single-phase SPH simulations that are applicable to different SPH schemes. However, the CSF approach relies strongly on the accurate evaluation of local interface normals and curvature, which can become challenging in the presence of large deformations, complex interface evolution, or insufficient particle resolution. As a result, its robustness may be compromised in certain droplet dynamics problems.

By contrast, the particle--particle interaction force (PIF) approach, which constitutes the second category, describes interfacial effects through pairwise forces between particles and thus provides a more direct representation of microscopic attractive and repulsive interactions. This characteristic makes it particularly attractive for modeling physical processes near interfaces and solid boundaries. For instance, Kordilla et al. \cite{kordilla2013smoothed} employed this approach to simulate droplet and film flows over a wide range of contact angles and Reynolds numbers on rock surfaces, with only the liquid and solid phases discretized by SPH particles. Nevertheless, a major limitation of early PIF-based SPH models was that the interaction-force parameters had to be calibrated empirically against the surface tension and static contact angle \cite{good1992contact}.  Tartakovsky and Panchenko \cite{tartakovsky2016pairwise} proposed a revised pairwise-force SPH formulation and derived explicit relationships between the interaction-force parameters and the surface tension and static contact angle for two- and three-phase flows in bounded domains. Despite this improvement, the method may still induce an additional virial pressure away from the interface, and its performance remains dependent on the specific functional form adopted for the pairwise interaction force. Therefore, the development of a more robust SPH model with fewer artificial parameters remains highly desirable.

To overcome the limitations of existing SPH approaches in modeling interfacial interactions, this work develops a physically grounded SPH framework for droplet dynamics from a microscopic perspective. The proposed method adopts a single-phase strategy to enhance computational efficiency, represents long-range interfacial interactions using the SPH kernel function, and incorporates short-range interactions into the pressure term. Furthermore, an explicit relationship between the intermolecular potential energy and the macroscopic surface tension coefficient is derived, which reduces the need for empirical parameter calibration. Overall, the proposed framework offers a physically consistent, robust, and efficient approach for simulating droplet dynamics on complex surfaces.

The main contributions are summarized as follows:
\begin{itemize}
    \item A single-phase SPH strategy is developed to improve the computational efficiency of droplet dynamics simulations.
    \item A physically grounded interfacial interaction model is proposed to describe both long-range and short-range effects with reduced empirical parameter dependence.
    \item Comprehensive wetting and impact simulations are performed to validate the proposed method and demonstrate its ability to capture complex dynamic wetting behaviors.
\end{itemize}

The remainder of this paper is organized as follows. Section~\ref{sec2} introduces the fundamentals of the SPH methodology. Section~\ref{sec3} presents a reformulation of the Young--Dupr\'e equation. Section~\ref{sec4} describes the governing equations and their SPH discretization. Section~\ref{sec5} provides several numerical examples and validation studies to demonstrate the accuracy and effectiveness of the proposed method. Finally, Section~\ref{sec6} summarizes the main conclusions of this work.

\section{Foundations of SPH Methodology}\label{sec2}

By the defining property of the Dirac delta function, one has
\begin{equation}\label{eq1}
  f(\mathbf{x})=\int_\Omega f(\mathbf{x}') \delta(\mathbf{x}-\mathbf{x}') \, \mathrm{d}\mathbf{x}'.
\end{equation}
Here, $\delta$ denotes the Dirac delta function, which satisfies
\[
\delta(\mathbf{x}-\mathbf{x}')=0 \quad \text{for } \mathbf{x}\neq \mathbf{x}', 
\qquad 
\int_\Omega \delta(\mathbf{x}-\mathbf{x}') \, \mathrm{d}\mathbf{x}' =1.
\]

In the SPH framework, the Dirac delta function $\delta(\mathbf{x}-\mathbf{x}')$ is replaced by a smoothing kernel function $W(\mathbf{x}-\mathbf{x}',h)$, where $h$ denotes the kernel support radius (or smoothing length). Similar to the Dirac delta function, the kernel function $W(\mathbf{x},h)$ has nonzero contribution only within a bounded support domain. Let $\|\cdot\|$ denote the Euclidean norm. In general, the kernel function is required to satisfy the following properties:
\begin{itemize}
  \item \textit{Normalization condition}:
  \(
  \int_\Omega W(\mathbf{x},h)\,\mathrm{d}\mathbf{x}=1;
  \)
  \item \textit{Symmetry property}:
  \(
  W(\mathbf{x},h)=W(-\mathbf{x},h);
  \)
  \item \textit{Compact support condition}:
  \(
  W(\mathbf{x},h)=0, \quad \text{for } \|\mathbf{x}\|\ge h.
  \)
\end{itemize}

The \textit{kernel approximation} (also referred to as the \textit{integral approximation}) of a scalar field $f(\boldsymbol{x})$ can be defined by
\begin{equation}\label{eq2}
    f_I(\boldsymbol{x}) := \int_\Omega f(\boldsymbol{x}') W(\boldsymbol{x} - \boldsymbol{x}', h) \, \mathrm{d}\boldsymbol{x}' \approx f(\boldsymbol{x}).
\end{equation}
Introducing the normalized distance
\[
q = \frac{|\boldsymbol{x} - \boldsymbol{x}'|}{h},
\]
we consider two widely used kernel functions.

\begin{enumerate}
\item \textbf{Cubic Spline Kernel} \cite{monaghan1992smoothed}:
\begin{align}
    W^{\mathrm{cs}}(q,h)  =
    \begin{cases}
    \sigma_{\mathrm{cs}}[1-6q^2(1-q)], & 0 \le q \le \frac12, \\[4pt]
    \sigma_{\mathrm{cs}}[2(1-q)^3], & \frac12 < q \le 1, \\[4pt]
    0, & q>1,
    \end{cases}
\end{align}
where $\sigma_{\mathrm{cs}} = \frac{8}{\pi h^3}$ in three dimensions.

\item \textbf{Wendland Quintic Kernel} \cite{gomez2010state}:
\begin{align}
    W^{\mathrm{wq}}(q,h) = 
    \begin{cases}
        \sigma_{\mathrm{wq}} (1-q)^4(4q+1), & 0 \leq q \leq 1, \\[4pt]
        0, & q > 1,
    \end{cases}
\end{align}
where $\sigma_{\mathrm{wq}} = \frac{21}{2\pi h^3}$ in three dimensions. The Wendland quintic kernel is often regarded as a favorable choice in terms of both computational accuracy and efficiency, since it provides relatively high-order interpolation while maintaining a computational cost comparable to that of lower-order kernels.

\end{enumerate}

Following Eq.~\eqref{eq2}, an approximation to the derivatives of a function can be derived by applying integration by parts:
\begin{equation}\label{eq3}
[D^\beta f]_I(\boldsymbol{x}) = -\int_\Omega f(\boldsymbol{x}') D^\beta_{\boldsymbol{x}'} W(\boldsymbol{x}-\boldsymbol{x}', h) \, \mathrm{d}\boldsymbol{x}',
\end{equation}
where $\beta$ is a multi-index representing the order of differentiation.

By discretizing the kernel approximation, one obtains the \textit{particle approximation} for scalar fields and for the divergence of vector fields:
\begin{align}
\langle f \rangle_i &= \sum_j \frac{m_j}{\rho_j} f_j W_{ij}^{\mathrm{wq}}, \label{PAf} \\
\langle \nabla \cdot \boldsymbol{f} \rangle_i &= -\sum_j \frac{m_j}{\rho_j} \boldsymbol{f}_j \cdot \nabla_j W_{ij}^{\mathrm{wq}}, \label{PADf}
\end{align}
where $f_j := f(\boldsymbol{x}_j)$, $W_{ij}^{\mathrm{wq}} := W^{\mathrm{wq}}(\boldsymbol{x}_i-\boldsymbol{x}_j,h)$, and $\nabla_j W_{ij}^{\mathrm{wq}} := \nabla_{\boldsymbol{x}_j}W^{\mathrm{wq}}(\boldsymbol{x}_i-\boldsymbol{x}_j,h)$. Here, $i$ denotes the target particle at which the approximation is evaluated, while $j$ denotes the neighboring particles located within the kernel support. The quantity $m_j/\rho_j$ approximates the volume associated with particle $j$.

In practical implementations, modified forms of Eq.~\eqref{PADf} are typically employed. In particular, when a derivative term is multiplied or divided by the density, it can be reformulated and incorporated into the summation operator. To this end, we make use of the following two identities \cite{monaghan1992smoothed}:
\begin{align}
[\rho\nabla\cdot\boldsymbol{f}](\boldsymbol{x}) &= \nabla\cdot(\rho \boldsymbol{f})(\boldsymbol{x}) - \boldsymbol{f}(\boldsymbol{x}) \cdot \nabla\rho(\boldsymbol{x}), \label{eq:div_identity} \\
\left[\frac{\nabla f}{\rho}\right](\boldsymbol{x}) &= \nabla\left(\frac{f}{\rho}\right)(\boldsymbol{x}) + \frac{f(\boldsymbol{x})}{\rho^2(\boldsymbol{x})}\nabla\rho(\boldsymbol{x}). \label{eq:grad_identity}
\end{align}
Combining these identities with Eq.~\eqref{PADf}, we obtain the following commonly used SPH discretizations for velocity and pressure:
\begin{align}
    \langle \rho\nabla\cdot\boldsymbol{u} \rangle_i &= -\sum_j m_j \boldsymbol{u}_{ij} \cdot \nabla_i W_{ij}^{\mathrm{wq}},  \\
    \left\langle \frac{\nabla p}{\rho} \right\rangle_i &= \sum_j m_j \left( \frac{p_i}{\rho_i^2} + \frac{p_j}{\rho_j^2} \right) \nabla_i W_{ij}^{\mathrm{wq}}.
\end{align}
Here, \(\boldsymbol{u}_{ij} := \boldsymbol{u}_{i} - \boldsymbol{u}_{j}\).
A widely used SPH discretization of the viscous term was proposed by Morris et al.~\cite{morris1997modeling}, namely,
\begin{align}
\left\langle \frac{\mu}{\rho} \nabla^2 \boldsymbol{u}\right\rangle_i
=
\sum_j \frac{ m_j (\mu_i+\mu_j)}{\rho_i\rho_j}
\frac{\boldsymbol{x}_{ij}\cdot\nabla_iW_{ij}^{\mathrm{wq}}}{|\boldsymbol{x}_{ij}|^2+(0.01 h)^2}\boldsymbol{u}_{ij}.
\end{align}

\section{Reformulation of Young--Dupr\'e Equation}\label{sec3}
The Young--Dupr\'e relation is a fundamental result in surface science that establishes a quantitative connection between the thermodynamic work of adhesion and key macroscopic wetting properties, in particular the equilibrium contact angle. It provides a theoretical criterion for characterizing the wetting behavior of a liquid on a solid surface and for determining whether partial or complete wetting occurs. However, the work of adhesion is not the most convenient quantity for direct implementation in SPH-based numerical simulations. Therefore, to facilitate numerical treatment and offer an alternative physical interpretation of contact-angle behavior, the Young--Dupr\'e relation is reformulated here in terms of intermolecular potential energy. 

\subsection{Interfacial Energy}
Consider a system of particles located at positions \(\bm r_i\). The pairwise interaction between particles \(i\) and \(j\) is described by the two-body potential \(\phi(|\bm r_i-\bm r_j|)\). The total interaction energy is
\begin{equation}
E=\sum_{i<j}\phi(|\bm r_i-\bm r_j|)
=\frac{1}{2}\sum_i\sum_{j\neq i}\phi(|\bm r_i-\bm r_j|),
\label{eq:discrete_energy}
\end{equation}
where the factor \(1/2\) removes the double counting of particle pairs.

To rewrite this expression in continuum form, we introduce the microscopic number density
\begin{equation}
n(\bm r):=\sum_i \delta(\bm r-\bm r_i).
\label{eq:number_density}
\end{equation}
Using \eqref{eq:number_density}, the double integral generates the full double sum,
\begin{equation}
\int_\Omega\int_\Omega n(\bm r)n(\bm r')\phi(|\bm r-\bm r'|)\,d\bm r\,d\bm r'
=
\sum_i\sum_j \phi(|\bm r_i-\bm r_j|),
\end{equation}
including the diagonal self-interaction terms \(i=j\), whose total contribution is
\begin{equation}
\sum_i \phi(0)=\phi(0)\int_\Omega n(\bm r)\,d\bm r.
\end{equation}
Therefore, subtracting the self-interaction terms and dividing by \(2\) to remove double counting, we obtain
\begin{equation}
E
=
\frac{1}{2}\int_\Omega\int_\Omega
n(\bm r)n(\bm r')\phi(|\bm r-\bm r'|)\,d\bm r\,d\bm r'
-\frac{\phi(0)}{2}\int_\Omega n(\bm r)\,d\bm r.
\label{eq:continuum_energy}
\end{equation}

Substituting \eqref{eq:number_density} into \eqref{eq:continuum_energy} gives
\begin{align}
E
&=
\frac{1}{2}\sum_i\sum_j
\int_\Omega\int_\Omega
\delta(\bm r-\bm r_i)\delta(\bm r'-\bm r_j)
\phi(|\bm r-\bm r'|)\,d\bm r\,d\bm r'
-\frac{1}{2}\sum_i\phi(0) \nonumber\\
&=
\frac{1}{2}\sum_i\sum_j \phi(|\bm r_i-\bm r_j|)
-\frac{1}{2}\sum_i\phi(0) \nonumber\\
&=
\frac{1}{2}\sum_i\sum_{j\neq i} \phi(|\bm r_i-\bm r_j|),
\end{align}
which coincides with \eqref{eq:discrete_energy}. Hence, \eqref{eq:discrete_energy} and \eqref{eq:continuum_energy} are equivalent.

For interfacial problems, however, the relevant energetic quantity is no longer the bulk interaction energy, but rather the interaction energy across two adjacent phases. We therefore consider the cross-interaction energy
\[
E_{\mathrm{surf}}
=
\int_{\Omega_1}\int_{\Omega_2}
n(\bm r)n(\bm r')
\phi(|\bm r-\bm r'|)\,d\bm r'\,d\bm r.
\]
For analytical convenience, we assume that the two phases are separated by a planar interface and define
\[
\Omega_1 := Q\times (0,\infty),
\qquad
\Omega_2 := \mathbb{R}^2\times (-\infty,0),
\]
where \(Q \subset \mathbb{R}^2\) denotes the interfacial region in the lateral directions. Under this assumption, the interface is flat and has area \(|Q|\). Then, the interfacial potential energy density associated with the interaction potential \(\phi(\cdot)\) is defined as
\begin{align}
\label{gammephi}
e_\phi:=\frac{E_{\mathrm{surf}}}{|Q|}
=
\frac{1}{|Q|}
\int_{Q\times(0,\infty)}
\int_{\mathbb R^2\times(-\infty,0)}
n(\bm r)n(\bm r')
\phi(|\bm r-\bm r'|)\,d\bm r'\,d\bm r.
\end{align}
In contrast to the interfacial energy density, the quantity \(e_\phi\) represents a potential-energy contribution and is therefore not necessarily positive. In particular, \(e_\phi<0\) whenever \(\phi(r)<0\) for all \(r>0\).

We now examine the relationship between the interfacial energy density and the interfacial potential function. In particular, the surface energy density is determined by the first absolute moment of the kernel. More precisely, the following theorem holds. 
\begin{theorem}[Moment representation of the surface energy density]\label{The1}
Let \(\bar{\phi}:[0,\infty)\to\mathbb R\) be measurable, and assume that all integrals below are finite. For any bounded measurable set \(Q\subset\mathbb R^2\), define
\begin{align}
\mathcal J_Q[\bar{\phi}]
:=
\frac{1}{|Q|}
\int_{Q\times(0,\infty)}
\int_{\mathbb R^2\times(-\infty,0)}
\bar{\phi}(|\bm r-\bm r'|)\,d\bm r'\,d\bm r.
\label{define_J}
\end{align}
Then
\begin{align}
\mathcal J_Q[\bar{\phi}]
=
\pi\int_0^\infty  r^3\bar{\phi}( r)\,d  r
=
\frac14\,M_1[\bar{\phi}],
\label{M1moment}
\end{align}
where
\[
M_1[\bar{\phi}]
:=
\int_{\mathbb R^3} |\hat{\bm r}|\,\bar{\phi}(|\hat{\bm r}|)\,d\hat{\bm r}
\]
is the first absolute moment of \(\bar{\phi}\). In particular, \(\mathcal J_Q[\bar{\phi}]\) is independent of \(Q\).
\end{theorem}

\begin{proof}
Write
\[
\bm r=(\bm y,z), \qquad \bm r'=(\bm y',z'),
\]
with \(\bm y,\bm y'\in\mathbb R^2\), \(z>0\), and \(z'<0\). Then \eqref{define_J} can be rewritten as 
\begin{align}
\mathcal J_Q[\bar{\phi}]
=
\frac{1}{|Q|}
\int_{Q}\int_0^\infty\int_{\mathbb R^2}\int_{-\infty}^0
\bar{\phi}\bigl(|(\bm y-\bm y',z-z')|\bigr)\,dz'\,d\bm y'\,dz\,d\bm y.
\end{align}
Since the integrand depends on \(\bm y\) and \(\bm y'\) only through their difference, let
\[
\bm\eta=\bm y-\bm y'.
\]
For each fixed \(\bm y\in Q\), this is a translation in \(\mathbb R^2\), so \(d\bm y'=d\bm\eta\). Therefore, one has
\begin{align}
\mathcal J_Q[\bar{\phi}]
&=
\frac{1}{|Q|}
\int_{Q}\int_0^\infty\int_{-\infty}^0\int_{\mathbb R^2}
\bar{\phi}\bigl(|(\bm\eta,z-z')|\bigr)\,d\bm\eta\,dz'\,dz\,d\bm y \notag\\
&=
\int_0^\infty\int_{-\infty}^0\int_{\mathbb R^2}
\bar{\phi}\bigl(|(\bm\eta,z-z')|\bigr)\,d\bm\eta\,dz'\,dz,
\end{align}
which already shows that \(\mathcal J_Q[\bar{\phi}]\) does not depend on \(Q\).

Next, introduce the change of variables
\[
t=z-z', \qquad s=z',
\]
so that \(z=t+s\), \(z'=s\), and the Jacobian is \(1\). The domain \(\{z>0,\ z'<0\}\) becomes
\(
\{t>0,\ -t<s<0\},
\)
which yields
\[
\mathcal J_Q[\bar{\phi}]
=
\int_0^\infty\int_{-t}^0\int_{\mathbb R^2}
\bar{\phi}\bigl(|(\bm\eta,t)|\bigr)\,d\bm\eta\,ds\,dt
=\int_0^\infty\int_{\mathbb R^2}
t\,\bar{\phi}\bigl(|(\bm\eta,t)|\bigr)\,d\bm\eta\,dt.
\]
Now write \(\hat{\bm r}=(\bm\eta,t)\in\mathbb R^3\), so that \(\rho_3=t\). Then, one has
\[
\mathcal J_Q[\bar{\phi}]
=
\int_{\{\hat{\bm r}\in\mathbb R^3:\rho_3>0\}}
\rho_3\,\bar{\phi}(|\hat{\bm r}|)\,d\hat{\bm r}.
\]

Using spherical coordinates in the upper half-space,
\[
\rho_3= r\cos\theta,
\qquad
d\hat{\bm r}= r^2\sin\theta\,d r\,d\theta\,d\psi,
\]
with \(\theta\in[0,\pi/2]\) and \(\bar{\phi}\in[0,2\pi]\), we obtain
\begin{align}
\mathcal J_Q[\bar{\phi}]
&=
\int_0^\infty\int_0^{2\pi}\int_0^{\pi/2}
 r\cos\theta\,\bar{\phi}( r)\, r^2\sin\theta
\,d\theta\,d\psi\,d r \notag\\
&=
\left(\int_0^{2\pi}d\psi\right)
\left(\int_0^{\pi/2}\sin\theta\cos\theta\,d\theta\right)
\int_0^\infty  r^3\bar{\phi}( r)\,d r =\pi\int_0^\infty r^3\bar{\phi}( r)\,d r.
\label{define_J_v1}
\end{align}
Note that
\begin{align}
M_1[\bar{\phi}]
=
\int_{\mathbb R^3} |\hat{\bm r}|\,\bar{\phi}(|\hat{\bm r}|)\,d\hat{\bm r}
=
4\pi\int_0^\infty  r^3\bar{\phi}( r)\,d r.
\label{define_M1}
\end{align}
Combining \eqref{define_J_v1} and \eqref{define_M1}, we obtain
\[
\mathcal J_Q[\bar{\phi}]
=
\frac14\,M_1[\bar{\phi}],
\]
which proves \eqref{M1moment}. 
This completes the proof.
\end{proof}

To obtain an explicit expression, we assume that the number density is piecewise constant in the two phases, i.e.,
\[
n(\bm r)=n_1 \quad \text{for } \bm r\in\Omega_1,
\qquad
n(\bm r)=n_2 \quad \text{for } \bm r\in\Omega_2.
\]
This assumption is reasonable for the present setting, as each phase is taken to be homogeneous away from the interface and the density variation across the interface is neglected at the level of the continuum description. Then, by Theorem \ref{The1}, the interfacial potential energy density associated with the interaction potential \(\phi\) can be expressed as
\begin{align}\label{interfatialenergy}
e_\phi
&=
\frac{1}{|Q|}
\int_{Q\times(0,\infty)}
\int_{\mathbb R^2\times(-\infty,0)}
n_1 n_2\,\phi(|\bm r-\bm r'|)\,d\bm r'\,d\bm r \notag\\
&=
n_1 n_2\,\mathcal J_Q[\phi]
=
\frac14\,M_1[n_1 n_2\phi].
\end{align}
In particular, for a planar interface, the interfacial potential energy density  admits an explicit representation in terms of the first absolute moment of the interaction kernel. Equivalently, it is completely determined by the quantity
\[
\int_{\mathbb R^3} |\hat{\bm r}|\,n_1 n_2 \phi(|\hat{\bm r}|)\,d\hat{\bm r}.
\]

\begin{remark}
While Theorem \ref{The1} assumes an infinite planar interface to derive the explicit parameter mapping, in practical SPH simulations, the localized support of
the smoothing kernel \(R_{cho}\) ensures that the planar approximation holds asymptotically away from the contact line, where the local interface curvature radius is
much larger than \(R_{cho}\).
\end{remark} 

\subsection{Contact Angle Analysis}

The equilibrium configuration at a three-phase contact line is governed by the Neumann triangle condition, which requires the vector sum of the interfacial tension forces to vanish:
\begin{equation}
    \vec{\gamma}_{12} + \vec{\gamma}_{23} + \vec{\gamma}_{13} = \vec{0},
    \label{eq:neumann_vector}
\end{equation}
where \(\vec{\gamma}_{ij}\) denotes the interfacial tension vector acting along the interface between phases \(i\) and \(j\), and \(\gamma_{ij} = |\vec{\gamma}_{ij}|\) is its Euclidean norm.

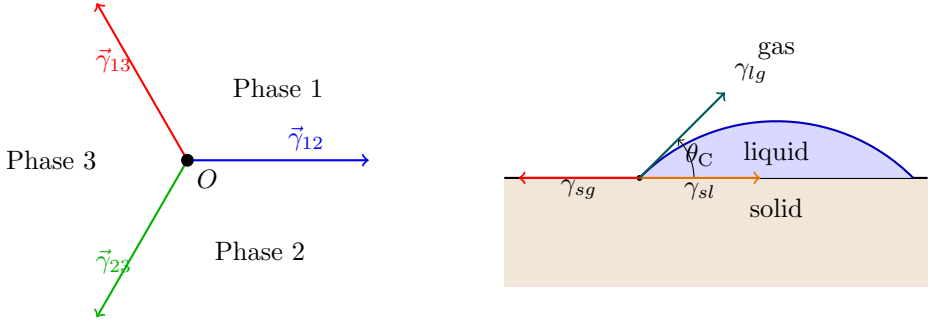
\begin{figure}[!htbp]
    \centering
    \begin{minipage}{0.45\textwidth}
        \centering
        \begin{tikzpicture}[scale=1.2]
            \coordinate (O) at (0,0);
            \coordinate (S12) at (2,0);
            \coordinate (S13) at (-1,1.732);
            \coordinate (S23) at (-1,-1.732);
            
            \draw[->, thick, blue] (O) -- (S12) node[midway, above right] {\(\vec{\gamma}_{12}\)};
            \draw[->, thick, red] (O) -- (S13) node[midway, above left] {\(\vec{\gamma}_{13}\)};
            \draw[->, thick, green!70!black] (O) -- (S23) node[midway, below left] {\(\vec{\gamma}_{23}\)};
            
            
            \fill (O) circle (2pt) node[below right] {\(O\)};
            
            \node at (1,0.8) {Phase 1};
            \node at (0.8,-1) {Phase 2};
            \node at (-1.5,0) {Phase 3};
        \end{tikzpicture}
    \end{minipage}
    \hfill
    \begin{minipage}{0.45\textwidth}
        \centering
        \begin{tikzpicture}[scale=0.8]

\def\R{3.2}
\def\thetaV{45}

\pgfmathsetmacro{\xa}{\R*cos(\thetaV)}
\pgfmathsetmacro{\yc}{-\R*sin(\thetaV)}

\coordinate (C) at (0,\yc);
\coordinate (L) at (-\xa,0);
\coordinate (R) at (\xa,0);

\fill[brown!20] (-4.5,-1.8) rectangle (2.5,0);

\draw[thick] (-4.5,0) -- (2.5,0);

\fill[blue!15] (L) arc[start angle={180-\thetaV}, end angle={\thetaV}, radius=\R] -- cycle;

\draw[thick,blue!70!black] (L) arc[start angle={180-\thetaV}, end angle={\thetaV}, radius=\R];

\node at (0,0.4) {liquid};
\node at (0,2.1) {gas};
\node at (0,-0.5) {solid};

\fill (L) circle (1.4pt);

\draw[dashed] (L) -- ++({1.8*cos(\thetaV)},{1.8*sin(\thetaV)});

\draw[->] ($(L)+(0.9,0)$) arc[start angle=0,end angle=\thetaV,radius=0.9];
\node at ($(L)+(1.0,0.38)$) {\(\theta_{\mathrm{C}}\)};

\draw[->,thick,red] (L) -- ++(-2.0,0);
\node[above] at ($(L)+(-1.0,-0.5)$) {\(\gamma_{sg}\)};

\draw[->,thick,orange!90!black] (L) -- ++(2.0,0);
\node[above] at ($(L)+(1.0,-0.5)$) {\(\gamma_{sl}\)};

\draw[->,thick,teal!70!black] (L) -- ++({2.0*cos(\thetaV)},{2.0*sin(\thetaV)});
\node[above right] at ($(L)+({2.0*cos(\thetaV)},{2.0*sin(\thetaV)})$) {\(\gamma_{lg}\)};

\end{tikzpicture}
    \end{minipage}
    \caption{Schematic illustrations of the Neumann triangle condition (left) and Young's equation (right).}
    \label{fig:two_figures}
\end{figure}

As illustrated in the left panel of Fig.~\ref{fig:two_figures}, the three interfacial tension vectors form a closed triangle, indicating local mechanical equilibrium at the contact point. The corresponding contact angles \(\theta_1\), \(\theta_2\), and \(\theta_3\) are therefore determined by the relative magnitudes of the interfacial tensions. Resolving Eq.~\eqref{eq:neumann_vector} into horizontal and vertical components gives
\begin{align}
    \gamma_{12} &= \gamma_{13}\cos(\pi-\theta_1) + \gamma_{23}\cos(\pi-\theta_2), 
    \label{eq:neumann_horiz} \\
    0 &= \gamma_{13}\sin(\pi-\theta_1) - \gamma_{23}\sin(\pi-\theta_2).
    \label{eq:neumann_vert}
\end{align}

When one of the three phases is replaced by a rigid solid, as shown in the right panel of Fig.~\ref{fig:two_figures}, the Neumann construction reduces to Young's equation,
\begin{equation}
    \gamma_{sg} = \gamma_{sl} + \gamma_{lg}\cos\theta_{\mathrm{C}},
    \label{eq:young}
\end{equation}
where \(\gamma_{sg}\), \(\gamma_{sl}\), and \(\gamma_{lg}\) are the solid--gas, solid--liquid, and liquid--gas interfacial tensions, respectively, and \(\theta_{\mathrm{C}}\) is the equilibrium contact angle.

To connect the equilibrium contact angle with interfacial energetics, we introduce the work of adhesion. Consider a virtual cleavage surface in a homogeneous liquid. The work required to separate the liquid across this surface per unit area equals twice the liquid surface tension,
\begin{align}
W_{ll} = 2\gamma_{l}.
\label{W2l}
\end{align}
For two dissimilar phases, the Dupr\'e relation \cite{good1992contact} reads
\begin{equation}
    W_{12} = \gamma_1 + \gamma_2 - \gamma_{12},
    \label{YDrelation}
\end{equation}
where \(\gamma_1\) and \(\gamma_2\) are the surface tensions of the two phases, \(\gamma_{12}\) is the interfacial tension, and \(W_{12}\) is the work of adhesion per unit area.

Neglecting adsorption at the solid--gas and liquid--gas interfaces, we take \(\gamma_s \approx \gamma_{sg}\) and \(\gamma_l \approx \gamma_{lg}\). Combining Young's equation with the Dupr\'e relation gives
\begin{align*}
    \gamma_{sg}
    &= \gamma_{sl} + \gamma_{lg}\cos\theta_{\mathrm{C}} \\
    &= (\gamma_s + \gamma_l - W_{sl}) + \gamma_{lg}\cos\theta_{\mathrm{C}} \\
    &\approx \gamma_{sg} + \gamma_{lg} - W_{sl} + \gamma_{lg}\cos\theta_{\mathrm{C}},
\end{align*}
and hence
\begin{equation}
    W_{sl} = \gamma_{lg}(1+\cos\theta_{\mathrm{C}}).
\end{equation}
This is the Young--Dupr\'e relation \cite{good1992contact,VANOSS200813}, or equivalently,
\begin{align}
    \cos\theta_{\mathrm{C}} = \frac{W_{sl}}{\gamma_{lg}} - 1.
    \label{young_2}
\end{align}

To relate the contact angle to microscopic interactions, let \(\phi_{sl}(r)\) denote the liquid--solid pair potential, and assume constant number densities \(n_l\) and \(n_s\) in the liquid and solid phases. Taking the fully separated state as the zero of energy and using the interfacial potential energy density defined in \eqref{gammephi}, we have
\begin{align}
W_{sl} = -e_{\phi_{sl}}, 
\qquad
W_{ll} = 2\gamma_{lg} = -e_{\phi_{ll}}.
\label{workandpotential}
\end{align}
Combining Eqs.~\eqref{young_2}, \eqref{workandpotential}, and \eqref{interfatialenergy}, we obtain
\begin{align}
    \cos\theta_C = \frac{2M_1[n_ln_s\phi_{sl}]}{M_1[n_ln_l\phi_{ll}]} - 1.
    \label{eq:contact_angle_micro}
\end{align}

Equation~\eqref{eq:contact_angle_micro} links the macroscopic contact angle to the microscopic liquid--solid interaction through the first moment of the interaction potential. Once the intermolecular potential and phase densities are specified, the contact angle can be estimated quantitatively.

\section{Governing Equations and SPH Discretization}\label{sec4}
Classical SPH methods can effectively capture the pressure and viscous terms in the Navier--Stokes equations. However, modeling droplet motion under different wettability conditions requires a more careful treatment of interfacial effects, particularly surface tension and liquid--solid adhesion. To this end, we introduce kernel-based pairwise potentials to characterize intermolecular interactions, so that the continuum formulation recovers the correct interfacial energy. This section presents the resulting model and its SPH discretization.

\subsection{Governing Equations}

Let \(\phi_{ll}(r)\) denote the liquid--liquid pair potential, where \(r = |\boldsymbol{x}-\boldsymbol{y}|\), and let \(n_l\) and \(n_s\) denote the number densities of the liquid and solid phases, respectively. The pairwise contribution to the liquid--liquid interaction force density at position \(\boldsymbol{x}\) due to a liquid particle located at \(\boldsymbol{y}\) is defined by
\begin{equation}
\boldsymbol{f}_{ll}^{\,p}(\boldsymbol{x},\boldsymbol{y})
=
-\nabla_{\boldsymbol{x}}
\left[
n_l^2 \phi_{ll}(|\boldsymbol{x}-\boldsymbol{y}|)
\right].
\end{equation}
Accordingly, the total liquid--liquid interaction force density acting on the fluid at \(\boldsymbol{x}\) is given by
\begin{equation}\label{f_ll}
\boldsymbol{f}_{ll}(\boldsymbol{x})
=
\int_{\Omega}
-\nabla_{\boldsymbol{x}}
\left[
n_l^2 \phi_{ll}(|\boldsymbol{x}-\boldsymbol{y}|)
\right]
\, d\boldsymbol{y},
\end{equation}
where \(\Omega\) denotes the liquid domain.

Similarly, let \(\phi_{ls}(r)\) denote the liquid--solid pair potential. Then the total adhesive force density exerted by the solid phase on the liquid at \(\boldsymbol{x}\) is expressed as
\begin{equation}\label{f_ls}
\boldsymbol{f}_{ls}(\boldsymbol{x})
=
\int_{\Omega_s}
-\nabla_{\boldsymbol{x}}
\left[
n_l n_s \phi_{ls}(|\boldsymbol{x}-\boldsymbol{y}|)
\right]
\, d\boldsymbol{y},
\end{equation}
where \(\Omega_s\) denotes the solid domain.

Although \(\phi_{ll}\) is introduced at the microscopic scale, its macroscopic effect corresponds to surface tension. Specifically, it satisfies
\begin{align}
M_1[n_l^2 \phi_{ll}] = -8\gamma_l,
\end{align}
which follows from \eqref{W2l}, \eqref{workandpotential}, and Theorem~\ref{The1}. Similarly, note that the work of adhesion at the liquid--solid interface is denoted by \(W_{sl}\), then
\begin{align}\label{lsCondition}
M_1[n_l n_s \phi_{ls}] = -4W_{sl}.
\end{align}

The droplet dynamics are governed by the Lagrangian forms of the continuity and momentum equations, together with an equation of state and the particle kinematic relation:
\begin{align}
\frac{d \rho}{d t} &= -\rho \nabla \cdot \boldsymbol{u},
&& \forall \boldsymbol{x} \in \Omega, \\
p &=
\frac{c^2 \rho_0}{k}
\left[
\left(\frac{\rho}{\rho_0}\right)^k - 1
\right],
&& \forall \boldsymbol{x} \in \Omega, \\
\rho \frac{d \boldsymbol{u}}{d t}
&=
-\nabla p
+ \mu \Delta \boldsymbol{u}
+ \rho \mathbf{g}
+ \boldsymbol{f}_{ll}(\boldsymbol{x})
+ \boldsymbol{f}_{ls}(\boldsymbol{x}),
&& \forall \boldsymbol{x} \in \Omega, \\
\frac{d \mathbf{x}}{d t} &= \boldsymbol{u},
&& \forall \boldsymbol{x} \in \Omega. 
\end{align}
Here, \(\rho\), \(p\), and \(\boldsymbol{u}\) denote the density, pressure, and velocity, respectively; \(\mu\) is the dynamic viscosity; \(\mathbf{g}\) is the gravitational acceleration; \(c\) is the artificial speed of sound; \(\rho_0\) is the reference density; and \(k\) is the exponent in the equation of state.

At the solid--liquid interface, a pressure boundary condition with a hydrostatic correction accounting for gravity is imposed, namely,
\[
\frac{\partial p}{\partial \mathbf{n}} = \rho \mathbf{g}\cdot \mathbf{n},
\qquad \forall \boldsymbol{x} \in \partial \Omega_s.
\]
This boundary condition also helps prevent fluid particles from penetrating the solid boundary.

\subsection{Nonlocal Interaction and SPH Discretization}

To facilitate numerical implementation, the microscopic interaction \(n_l^2\phi_{ll}\) is replaced by an effective kernel-based cohesion potential of the form
\[
\Phi_{ll}(\bm r; R_{\mathrm{cho}})
:= C(R_{\mathrm{cho}})\, W(\bm r; R_{\mathrm{cho}}),
\]
where \(W(\bm r; R_{\mathrm{cho}})\) is a smoothing kernel with support radius \(R_{\mathrm{cho}}\), and \(C(R_{\mathrm{cho}})\) is a coefficient to be determined. The effective potential is required to preserve the interfacial energy density, i.e.,
\[
e_{\phi_{ll}}
=
\frac{1}{4} M_1[n_1 n_2 \phi_{ll}]
=
\frac{1}{4} M_1[\Phi_{ll}(\bm r; R_{\mathrm{cho}})].
\]
This condition uniquely determines \(C(R_{\mathrm{cho}})\), yielding
\begin{equation}
\Phi_{ll}(\bm r; R_{\mathrm{cho}})
=
-\frac{8\gamma_l}{M_1[W(\cdot; R_{\mathrm{cho}})]}
\, W(\bm r; R_{\mathrm{cho}}).
\label{eq:effective_potential}
\end{equation}
Hence, \(\Phi_{ll}\) preserves the first absolute moment and thus the interfacial energy per unit area.

In the present work, the cubic spline kernel is adopted to model the nonlocal interactions. We denote
\[
W^{\mathrm{cs}}_{R_{\mathrm{cho}}}(\boldsymbol{r})
:= W^{\mathrm{cs}}(\boldsymbol{r}; R_{\mathrm{cho}}),
\]
with first absolute moment
\[
M_1[W^{\mathrm{cs}}_{R_{\mathrm{cho}}}]
:=
\int_{\mathbb R^3} |\bm r|\, W^{\mathrm{cs}}(\boldsymbol{r}; R_{\mathrm{cho}})\, d\bm r.
\]

Let \(\boldsymbol{x}\) and \(\boldsymbol{y}\) be the positions of two finite-sized particle clusters. The effective liquid--liquid interaction potential is defined as
\begin{equation}
\Phi_{ll}(|\boldsymbol{x}-\boldsymbol{y}|; R_{\mathrm{cho}})
:=
-\frac{8\gamma_l}{M_1[W^{\mathrm{cs}}_{R_{\mathrm{cho}}}]}
\, W^{\mathrm{cs}}(\boldsymbol{x}-\boldsymbol{y}; R_{\mathrm{cho}}),
\end{equation}
where \(\gamma_l\) is the liquid surface tension coefficient. The corresponding pairwise liquid--liquid interaction force density is obtained from the negative gradient of the potential:
\begin{equation}
\hat{\boldsymbol{f}}_{ll}^{\,p}(\boldsymbol{x}, \boldsymbol{y})
=
-\nabla_{\boldsymbol{x}} \Phi_{ll}(|\boldsymbol{x}-\boldsymbol{y}|; R_{\mathrm{cho}})
=
\frac{8\gamma_l}{M_1[W^{\mathrm{cs}}_{R_{\mathrm{cho}}}]}
\nabla_{\boldsymbol{x}} W^{\mathrm{cs}}(\boldsymbol{x}-\boldsymbol{y}; R_{\mathrm{cho}}).
\label{eq:f_att_pair}
\end{equation}

Integrating the pairwise interaction over the liquid domain gives the SPH kernel approximation of the total liquid--liquid interaction force density:
\begin{equation}\label{hat_f_ll}
\hat{\boldsymbol{f}}_{ll}(\boldsymbol{x})
=
\frac{8\gamma_l}{M_1[W^{\mathrm{cs}}_{R_{\mathrm{cho}}}]}
\int_{\Omega}
\nabla_{\boldsymbol{x}} W^{\mathrm{cs}}(\boldsymbol{x}-\boldsymbol{y}; R_{\mathrm{cho}})
\, d\boldsymbol{y},
\end{equation}
where \(\Omega\) denotes the liquid domain. Similarly, the liquid--solid adhesive force density is approximated by
\begin{equation}\label{hat_f_ls}
\hat{\boldsymbol{f}}_{ls}(\boldsymbol{x})
=
\frac{4 W_{sl}}{M_1[W^{\mathrm{cs}}_{R_{\mathrm{cho}}}]}
\int_{\Omega_s}
\nabla_{\boldsymbol{x}} W^{\mathrm{cs}}(\boldsymbol{x}-\boldsymbol{y}; R_{\mathrm{cho}})
\, d\boldsymbol{y},
\end{equation}
where \(\Omega_s\) denotes the solid domain and \(W_{sl}\) is the work of adhesion between the liquid and solid phases.

Applying the SPH discretization to \eqref{hat_f_ll} and \eqref{hat_f_ls} leads to the following semi-discrete formulation. Let \(N_f\) and \(N_s\) denote the numbers of fluid and solid particles, respectively. For each fluid particle \(i\in\{1,\dots,N_f\}\), we solve
\begin{align}
\frac{d \rho_i}{d t}
&=
\sum_{j=1}^{N_f+N_s}
m_j \,\boldsymbol{u}_{ij}\cdot\nabla_i W_{ij}^{\mathrm{wq}},
\qquad
p_i
=
\frac{c^2\rho_0}{k}
\left[
\left(\frac{\rho_i}{\rho_0}\right)^k - 1
\right],
\\
\frac{d \boldsymbol{u}_i}{d t}
&=
-\sum_{j=1}^{N_f+N_s}
m_j
\left(
\frac{p_i}{\rho_i^2}
+\frac{p_j}{\rho_j^2}
+\Pi_{ij}^{\mathrm{art}}
\right)
\nabla_i W_{ij}^{\mathrm{wq}}\nonumber\\
&\quad
+\frac{8\gamma_l}{M_1[W^{\mathrm{cs}}_{R_{\mathrm{cho}}}]}
\sum_{j=1}^{N_f}
\frac{m_j}{\rho_i\rho_j}
\nabla_i W_{ij}^{\mathrm{cs},R_{\mathrm{cho}}}
+\frac{4W_{sl}}{M_1[W^{\mathrm{cs}}_{R_{\mathrm{cho}}}]}
\sum_{k=1}^{N_s}
\frac{m_k}{\rho_i\rho_k}
\nabla_i W_{ik}^{\mathrm{cs},R_{\mathrm{cho}}}\nonumber\\
&\quad
+\sum_{j=1}^{N_f}
\frac{m_j(\mu_i+\mu_j)}{\rho_i\rho_j}
\frac{\mathbf{x}_{ij}\cdot\nabla_i W_{ij}^{\mathrm{wq}}}
{|\mathbf{x}_{ij}|^2+0.01h_{ij}^2}
\boldsymbol{u}_{ij}
+\mathbf{g}_i,
\\
\frac{d \mathbf{x}_i}{d t}
&=
\boldsymbol{u}_i-\epsilon\sum_{j=1}^{N_f}
m_j\frac{\boldsymbol{u}_{ij}}{\bar{\rho}_{ij}} W_{ij}^{\mathrm{wq}}.
\end{align}
Here, \(\epsilon=0.5\) \cite{monaghan1994simulating}, 
\(\boldsymbol{u}_{ij}=\boldsymbol{u}_i-\boldsymbol{u}_j\), 
\(\bar{\rho}_{ij}=(\rho_i+\rho_j)/2\), 
and \(\mathbf{x}_{ij}=\mathbf{x}_i-\mathbf{x}_j\). 
The artificial viscosity \(\Pi_{ij}^{\mathrm{art}}\) is taken in the Monaghan form \cite{monaghan2005smoothed}:
\begin{align}
\Pi_{ab}^{\mathrm{art}}
=
\begin{cases}
\dfrac{-\alpha c\,\mu_{ab}}{\bar{\rho}_{ab}}, & \boldsymbol{u}_{ab}\cdot\boldsymbol{x}_{ab}<0,\\[1ex]
0, & \boldsymbol{u}_{ab}\cdot\boldsymbol{x}_{ab}\geq 0,
\end{cases}
\qquad
\mu_{ab}
=
\frac{h\,\boldsymbol{u}_{ab}\cdot\boldsymbol{x}_{ab}}
{|\boldsymbol{x}_{ab}|^{2}+0.01h^{2}}.
\end{align}
Here, \(\alpha\) is the artificial viscosity coefficient and \(c\) is the numerical speed of sound. In the momentum equation, the second and third terms on the right-hand side correspond to the liquid--liquid cohesive force and the liquid--solid adhesive force, respectively, while the other terms represent pressure, viscous diffusion, and gravity.

At the solid--liquid interface, the pressure of a solid particle is evaluated from its neighboring fluid particles using a kernel-weighted interpolation with a hydrostatic correction \cite{adami2012generalized}:
\begin{align}
p_k
=
\frac{
\sum_i^{N_f} p_i W_{ki}^{\mathrm{wq}}
+
\mathbf{g}\cdot
\sum_i^{N_f} \rho_i \mathbf{x}_{ki} W_{ki}^{\mathrm{wq}}
}{
\sum_i^{N_f} W_{ki}^{\mathrm{wq}}
},
\end{align}
where \(k\) and \(i\) denote the solid and fluid particles, respectively, and \(\mathbf{x}_{ki}=\mathbf{x}_k-\mathbf{x}_i\). The second term in the numerator accounts for the hydrostatic pressure variation induced by gravity and improves the pressure approximation near the solid boundary.

Since the SPH equations are integrated explicitly in time, the time step is restricted by the CFL condition. In this work, it is chosen as
\begin{align}
\Delta t = 0.1\,\frac{h}{c},
\end{align}
where \(h\) is the smoothing length. The parameter \(c\) in the equation of state denotes the numerical speed of sound and is chosen to satisfy
\[
c \ge \max\!\left(1~\mathrm{m/s},\, 10\,U_{\max}\right),
\]
where \(U_{\max}\) is the characteristic maximum fluid velocity.

To establish a simple relation between the liquid--solid interaction strength and the equilibrium contact angle, we introduce the dimensionless adhesion coefficient
\begin{equation}
\alpha_{\mathrm{adh}}
:=
\frac{W_{sl}}{2\gamma_l}.
\end{equation}
Substituting this definition into Eq.~\eqref{young_2} gives
\begin{equation}
\cos\theta_C
=
2\alpha_{\mathrm{adh}} - 1.
\label{eq:simple}
\end{equation}
This relation shows that the cosine of the equilibrium contact angle depends linearly on the adhesion coefficient \(\alpha_{\mathrm{adh}}\).

\begin{remark}
In realistic fluid--fluid interactions, the intermolecular force typically comprises both a long-range attractive component and a short-range repulsive component, with the latter preventing particles from approaching each other at excessively small distances. Their combined effect gives rise to a potential profile qualitatively similar to the Lennard--Jones potential (see Fig.~\ref{fig:Lennard-Jones}). In the present model, the SPH kernel \(W(r,h)\) is used to approximate the attractive tail of the interaction potential, i.e., the portion extending from the potential minimum toward the far field (see the blue solid line). By contrast, the short-range repulsive component acts only over a much smaller length scale and is represented by the pressure term in the present formulation. Because its contribution to the first absolute moment is negligible, it is not included in the present calibration.
\end{remark}

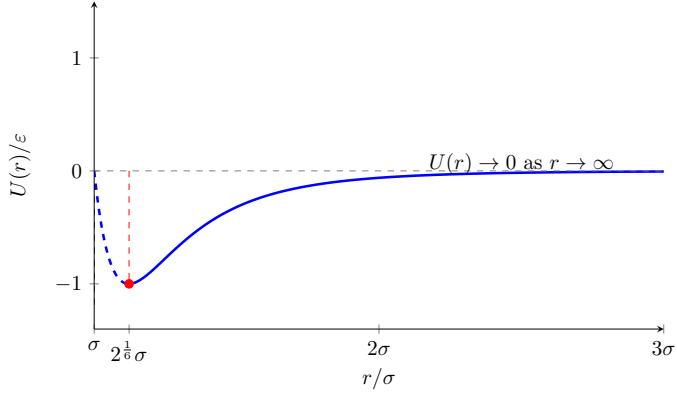
\begin{figure}
  \centering
  \begin{tikzpicture}[scale=0.8]
\begin{axis}[
    width=11cm,
    height=7cm,
    axis lines=left,
    xlabel={$r/\sigma$},
    ylabel={$U(r)/\varepsilon$},
    xmin=1, xmax=3.0,
    ymin=-1.4, ymax=1.5,
    xtick={1,1.12246,2,3},
    xticklabels={$\sigma$,$2^{\frac{1}{6}}\sigma$,$2\sigma$,$3\sigma$},
    ytick={-1,0,1,2},
    enlargelimits=false,
    clip=true
]

\addplot[
    domain=1:1.12246,
    samples=150,
    very thick,
    blue,
    dashed
]
{4*((1/x)^12 - (1/x)^6)};

\addplot[
    domain=1.12246:3.0,
    samples=150,
    very thick,
    blue
]
{4*((1/x)^12 - (1/x)^6)};

\addplot[
    dashed,
    gray
]
{0};

\addplot[
    only marks,
    mark=*,
    mark size=2pt,
    red
]
coordinates {(1.12246,-1)};

\addplot[
    dashed,
    black
]
coordinates {(1,0) (1,-1.2)};

\addplot[
    dashed,
    red
]
coordinates {(1.12246,0) (1.12246,-1)};

\node[above right] at (axis cs:2.15,-0.1) {$U(r)\to 0$ as $r\to\infty$};

\end{axis}
\end{tikzpicture}
  \caption{Schematic of the Lennard--Jones potential. In the present model, the attractive tail from the potential minimum to the far field is approximated by the effective kernel-based cohesion potential.}
  \label{fig:Lennard-Jones}
\end{figure}

\section{Numerical Verification}\label{sec5}
In this section, several representative numerical examples are presented to validate the theoretical results and demonstrate the scalability of the proposed model. All numerical simulations were performed using the PySPH library \cite{ramachandran2021a}. The validation videos and the corresponding source code are publicly available at \href{https://doi.org/10.5281/zenodo.21404140}{\color{blue}DOI: 10.5281/zenodo.21404140}.

\subsection{Validation of the Contact-Angle Relation and Equilibrium Droplet Shapes}

To validate the theoretical relation between the adhesion coefficient and the equilibrium contact angle, we perform a series of static droplet simulations under zero gravity. The parameters are set to \(R_{\mathrm{cho}} = 0.20\,\mathrm{mm}, h = 0.11\,\mathrm{mm}\) and an average particle spacing of \(\Delta x = 0.05\,\mathrm{mm}\). For each prescribed value of \(\alpha_{\mathrm{adh}}\), the equilibrium droplet profile is measured and the corresponding contact angle is extracted.

Figure~\ref{fig:fitting_comparison} compares the simulation results, expressed in terms of \(\cos(\theta_C)\), with the theoretical prediction \eqref{eq:simple}. 
Overall, the numerical results follow the theoretical linear relation closely over the range of \(\alpha_{\mathrm{adh}}\) considered. A least-squares fit to the numerical data yields
\[
\cos(\theta_C) = 2.081\,\alpha_{\mathrm{adh}} - 1.104,
\]
which is in good agreement with the theoretical slope and intercept in \eqref{eq:simple}. 
To further quantify the agreement, the coefficient of determination \(R^2\), the adjusted \(R^2\), the root-mean-square error (RMSE), and the mean absolute error (MAE) are evaluated. The results are summarized in Table~\ref{tab:fitting_quality}. The fact that the \(R^2\) values are close to 1, together with the small error metrics, indicates that the numerical simulations agree closely with the theoretical relation for the contact angle.

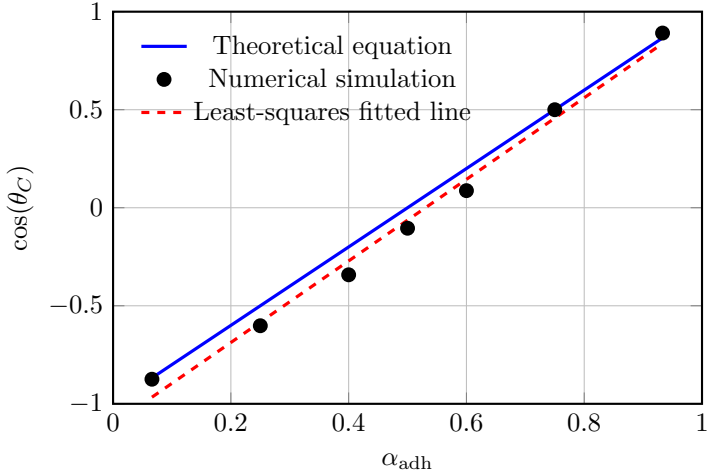
\begin{figure}[htbp]
\centering
\begin{tikzpicture}
\begin{axis}[
    width=0.72\textwidth,
    height=0.52\textwidth,
    xlabel={\(\alpha_{\mathrm{adh}}\)},
    ylabel={\(\cos(\theta_C)\)},
    xmin=0, xmax=1,
    ymin=-1, ymax=1,
    grid=both,
    thick,
    legend style={
        at={(0.03,0.97)},
        anchor=north west,
        draw=none,
        fill=none
    },
    tick label style={font=\normalsize},
    label style={font=\normalsize},
]

\addplot[
    blue,
    line width=1.2pt,
    domain=0.066:0.933,
    samples=200
]
{2*x - 1};
\addlegendentry{Theoretical equation}

\addplot[
    only marks,
    mark=*,
    mark size=2.4pt,
    black
]
coordinates {
    (0.066,-0.8746)
    (0.250,-0.6018) 
    (0.400,-0.3420)
    (0.500,-0.1045)
    (0.600, 0.0872)
    (0.750,0.5000)
    (0.933,0.8910) 
};
\addlegendentry{Numerical simulation}

\addplot[
    red,
    dashed,
    line width=1.2pt,
    domain=0.066:0.933,
    samples=200
]
{2.080655*x  -1.103576};
\addlegendentry{Least-squares fitted line}

\end{axis}
\end{tikzpicture}
\caption{Comparison between the theoretical contact-angle relation, the least-squares fitted line, and the transformed numerical data.}
\label{fig:fitting_comparison}
\end{figure}
\begin{table}[htbp]
\centering
\caption{Comparison of fitting quality between the least-squares fitted line and the theoretical equation.}
\label{tab:fitting_quality}
\begin{tabular}{lcc}
\toprule
Metric & Least-squares fitted line & Theoretical equation \\
\midrule
Equation & \(\cos(\theta_C) = 2.081\,\alpha_{\mathrm{adh}} -1.104\) & \(\cos(\theta_C) = 2\,\alpha_{\mathrm{adh}} - 1\) \\
\(R^2\) & 0.989701 & 0.975918 \\
Adjusted \(R^2\) & 0.987641 & 0.971102 \\
RMSE & \(5.789650\times10^{-2}\) & \(8.853183\times10^{-2}\) \\
MAE & \(5.372844\times10^{-2}\) & \(7.040487\times10^{-2}\) \\
\bottomrule
\end{tabular}
\end{table}

To gain a more intuitive understanding of the simulated contact angles, we further investigate the equilibrium shapes of droplets on a solid substrate. Figure~\ref{fig:contact_angles_simple} shows three representative equilibrium droplet configurations corresponding to different wetting regimes. 
When \(W_{sl} = (1+\sqrt{3}/2)\gamma_l\), the equilibrium contact angle is \(\theta_C=\pi/6\), representing a strongly hydrophilic surface; see Fig.~\ref{fig:contact_angles_simple}(a). When \(W_{sl}=\gamma_l\), one obtains \(\theta_C=\pi/2\), corresponding to a neutral wetting state, as shown in Fig.~\ref{fig:contact_angles_simple}(b). When \(W_{sl}=(1-\sqrt{3}/2)\gamma_l\), the equilibrium contact angle becomes \(\theta_C=5\pi/6\), indicating a hydrophobic surface; see Fig.~\ref{fig:contact_angles_simple}(c). These results demonstrate that the proposed model can reproduce a broad range of prescribed equilibrium contact angles by tuning the liquid--solid work of adhesion.

\begin{figure}[!htbp]
    \centering
    \includegraphics[width=0.31\textwidth]{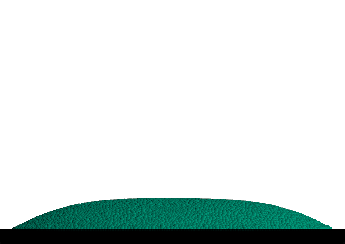}
    \hfill
    \includegraphics[width=0.31\textwidth]{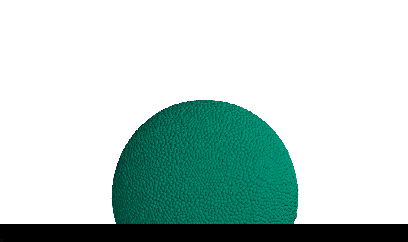}
    \hfill
    \includegraphics[width=0.31\textwidth]{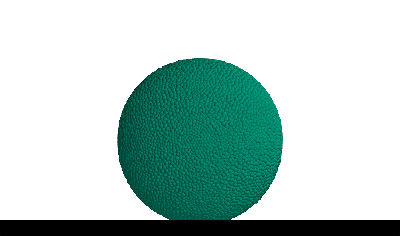}

    \medskip

    \parbox{0.3\textwidth}{\centering (a) \(W_{sl}=(1+\sqrt{3}/2)\gamma_l\)}
    \hfill
    \parbox{0.3\textwidth}{\centering (b) \(W_{sl}=\gamma_l\)}
    \hfill
    \parbox{0.3\textwidth}{\centering (c) \(W_{sl}=(1-\sqrt{3}/2)\gamma_l\)}

    \caption{Equilibrium droplet shapes for different values of the liquid--solid work of adhesion \(W_{sl}\). The corresponding contact angles are (a) \(\theta_C=\pi/6\), (b) \(\theta_C=\pi/2\), and (c) \(\theta_C=5\pi/6\).}
    \label{fig:contact_angles_simple}
\end{figure}

The capability of the present model to capture equilibrium morphologies governed by inter-liquid adhesion is further demonstrated in Fig.~\ref{fig:time_dependent_simple}. Two droplets of identical size and equal surface tension, \(\gamma_{l_1}=\gamma_{l_2}\), are considered under different values of the adhesion work \(W_{l_1l_2}\) between the two liquid phases. For this symmetric configuration, Eqs.~\eqref{eq:neumann_horiz}--\eqref{eq:neumann_vert} simplify to
\begin{align}
    \cos\theta_{\mathrm{eq}}=\frac{W_{l_1l_2}}{2\gamma_{l_1}}-1,
\end{align}
where \(\theta_{\mathrm{eq}}\) denotes the equilibrium contact angle.

\begin{figure}[!htbp]
    \centering
    \includegraphics[width=0.32\textwidth]{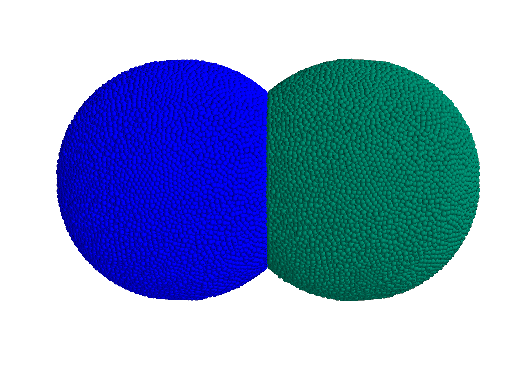}
    \includegraphics[width=0.32\textwidth]{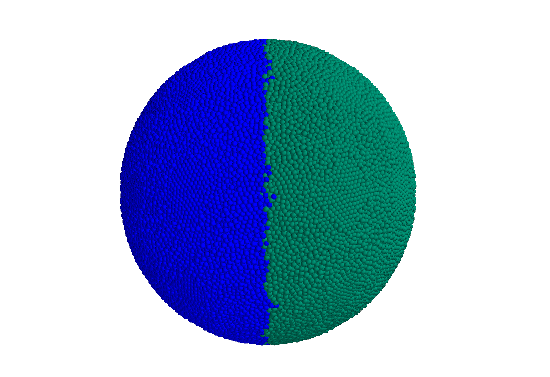}

    \medskip

    \parbox{0.3\textwidth}{\centering (a) \(W_{l_1l_2}=\gamma_{l_1}\)} 
    \parbox{0.3\textwidth}{\centering (b) \(W_{l_1l_2}=2\gamma_{l_1}\)}

    \caption{Equilibrium shapes of two droplets of the same size under different liquid--liquid adhesion energies. The corresponding contact angles are (a) \(\theta_C=2\pi/3\) and (b) \(\theta_C=\pi/2\).}
    \label{fig:time_dependent_simple}
\end{figure}
As shown in Fig.~\ref{fig:time_dependent_simple}, when \(W_{l_1l_2}=\gamma_{l_1}\), the equilibrium contact angle is \(\theta_{\mathrm{eq}}=2\pi/3\), and the two droplets appear as adjoining caps larger than hemispheres. When the adhesion work is increased to \(W_{l_1l_2}=2\gamma_{l_1}\), the equilibrium contact angle decreases to \(\theta_{\mathrm{eq}}=\pi/2\), and the droplets each take on a hemispherical shape, merging into a smooth and rounded overall configuration. This comparison indicates that stronger adhesion between the two liquid phases enhances their mutual affinity and reduces the equilibrium contact angle. The good agreement between the theoretical relation and the simulated droplet profiles confirms that the present model accurately captures the role of inter-liquid adhesion in determining the equilibrium morphology.

\subsection{Nonsteady Confirmation of Tanner's Law at the Complete Wetting Limit}
To verify Tanner's Law numerically, we simulated the spontaneous spreading of a droplet on a perfectly wetting solid surface. 
For this extreme state of adhesion where $\alpha_\text{adh} =1$, the strong solid-liquid attraction yields a hydrophilic surface with a contact angle of $\theta = 0$. 
Figure~\ref{fig:tanner_verification} presents the temporal evolution of the spreading radius $R$ on a log-log scale, with time $t$ measured in milliseconds and diameter in millimeters. 
According to Tanner's law, in the capillary regime where viscous forces dominate inertia, the spreading radius should follow the power-law relationship \(R(t)\propto t^{1/10}\) \cite{tanner1979spreading,wang2024role}. To test this prediction, the numerical data were fitted using the function
\begin{equation}
    R(t) = 2 \times (t - t_{\text{min}})^{0.1},
    \label{eq:tanner_fit}
\end{equation}
where \(t_{\text{min}}\) is a time-offset parameter introduced to account for the finite onset time of the Tanner-regime spreading.

\begin{figure}[!htbp]
    \centering
    \includegraphics[width=0.7\textwidth]{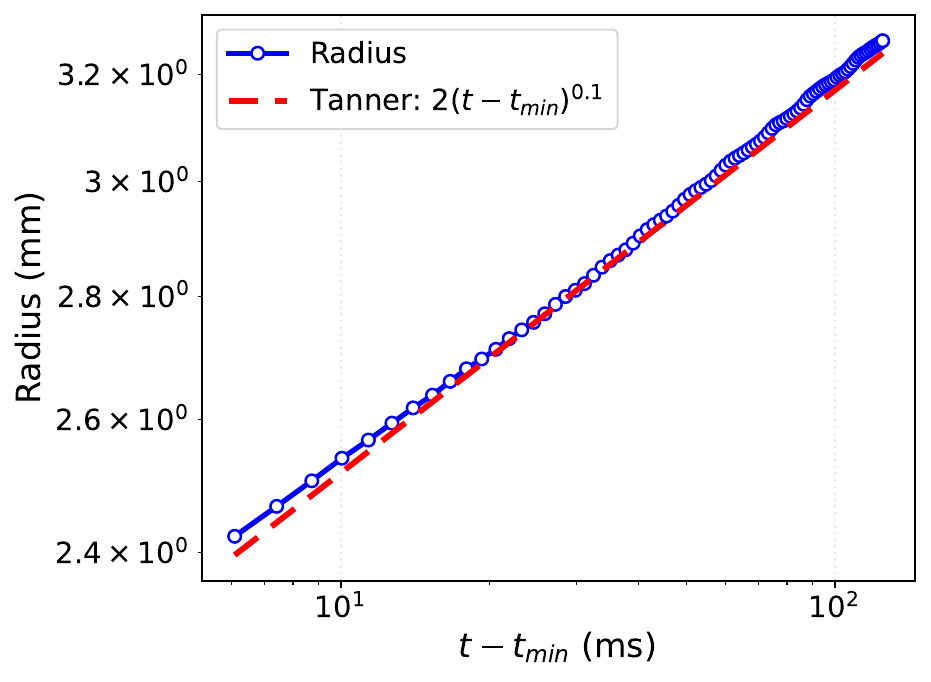}
    \caption{Log-log plot of droplet spreading radius as a function of time with $t_{min}=15$ms.}
    \label{fig:tanner_verification}
\end{figure}
As shown in Fig.~\ref{fig:tanner_verification}, the numerical results (represented by solid symbols) remain close to the fitted curve (dashed line) over nearly two decades in time, from approximately \(10\) ms to \(1000\) ms. The fitted exponent of \(0.1\) agrees well with the theoretical exponent predicted by Tanner's law. These results suggest that the numerical model captures the main features of contact-line dynamics and is consistent with Tanner's law in the droplet-spreading regime.

\subsection{Numerical Verification of Rolling Droplet Rebound}

To further validate the capability of the proposed SPH model in predicting complex droplet--surface interactions, we simulate the recently reported rolling droplet rebound on a patterned wettability (PW) surface \cite{zhao2025limit}. In this newly identified rebound mode, the droplet rolls rapidly along the substrate during recoil, leading to an apparent rebound angle approaching zero, which represents the theoretical lower limit of droplet rebound angles. Owing to the coexistence of superhydrophilic and superhydrophobic regions, this phenomenon involves strong wettability contrast, contact-line pinning, asymmetric capillary retraction, and directional momentum redistribution, thereby providing a stringent test for the present model.

The PW surface consists of a superhydrophilic (SHL) arc integrated onto a superhydrophobic (SHB) substrate. The SHL region exhibits an equilibrium contact angle below \(3^\circ\), whereas the SHB background has a contact angle of approximately \(158^\circ\). As shown in Fig.~\ref{PW}, the SHL arc has a line width of \(L = 200\,\mu\mathrm{m}\), an opening angle of \(\alpha = \pi/5\), and a radius of \(R = 2.5\,\mathrm{mm}\). This radius is chosen to be close to the maximum spreading radius of the impacting droplet, such that the receding liquid rim remains pinned by the SHL arc during recoil and is therefore subjected to enhanced adhesive resistance.

\begin{figure}[htbp]
  \centering
  \includegraphics[width=0.5\textwidth]{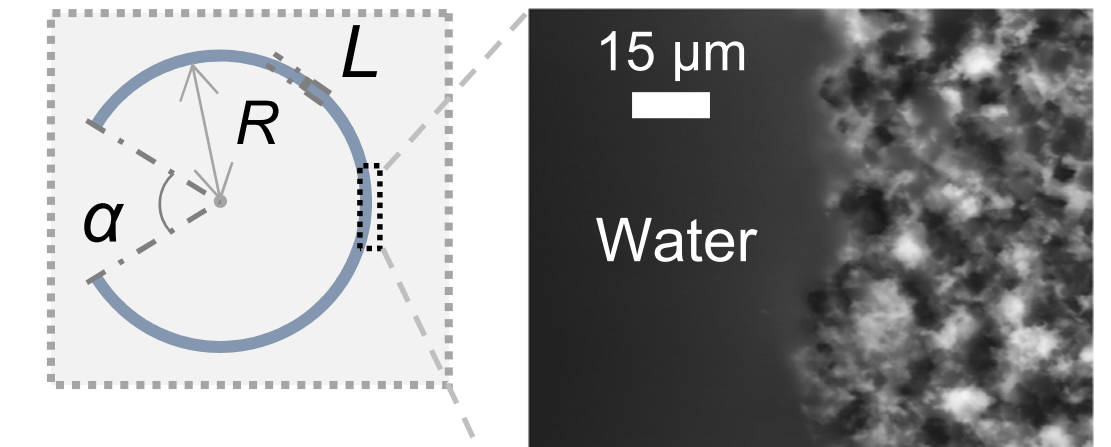}
  \caption{Schematic and microscopic characterization of PW surface. The surface consists of a SHL arc embedded in a SHB background.}
  \label{PW}
\end{figure}

The droplet impact condition is characterized by the Weber number,
\begin{equation}
    We = \frac{\rho v_0^2 R_0}{\gamma},
\end{equation}
where \(\rho\), \(v_0\), \(R_0\), and \(\gamma\) denote the liquid density, impact velocity, initial droplet radius, and surface tension, respectively. In the present case, a water droplet with \(We = 32.8\) impacts the PW surface. After impact, the droplet spreads rapidly and reaches its maximum lateral extension at approximately \(t = 2.2\,\mathrm{ms}\), after which the liquid film begins to retract and peel off from the substrate.

\begin{figure}[!htbp]
  \centering
  \includegraphics[width=0.78\textwidth]{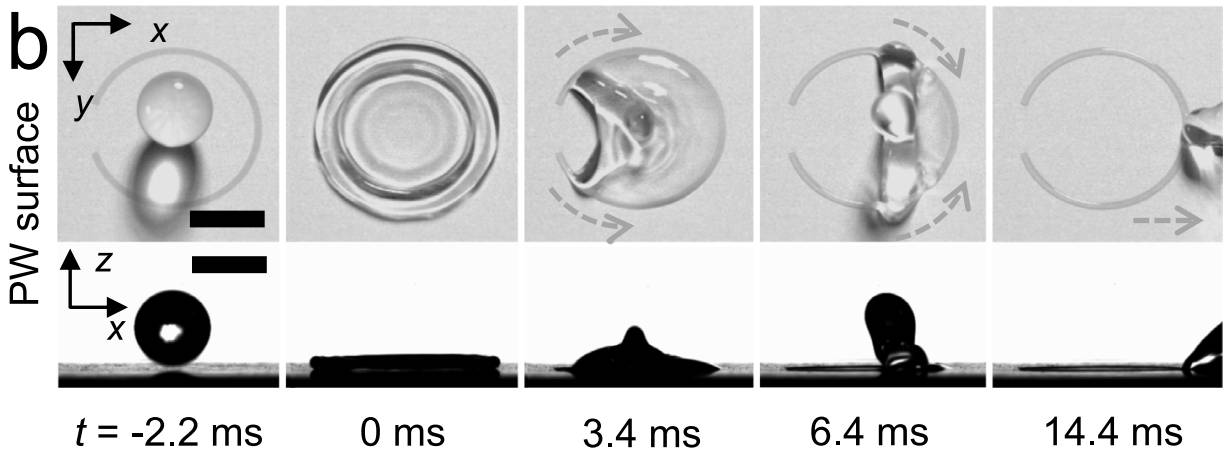}

  \vspace{0.6em}

  \includegraphics[width=0.15\textwidth]{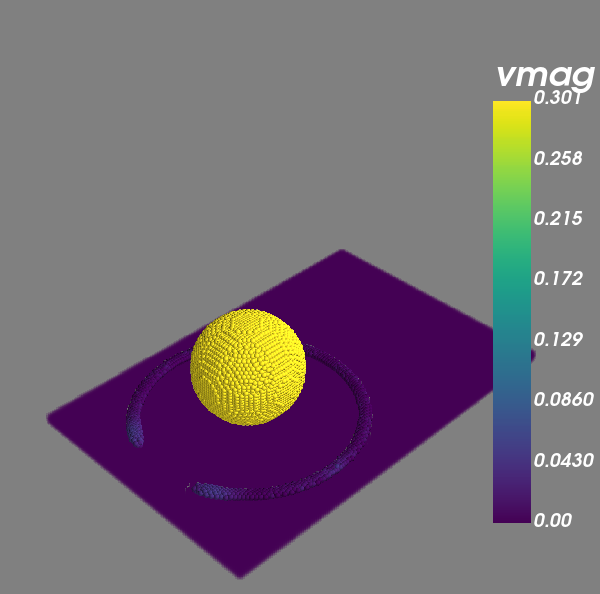}
  \includegraphics[width=0.15\textwidth]{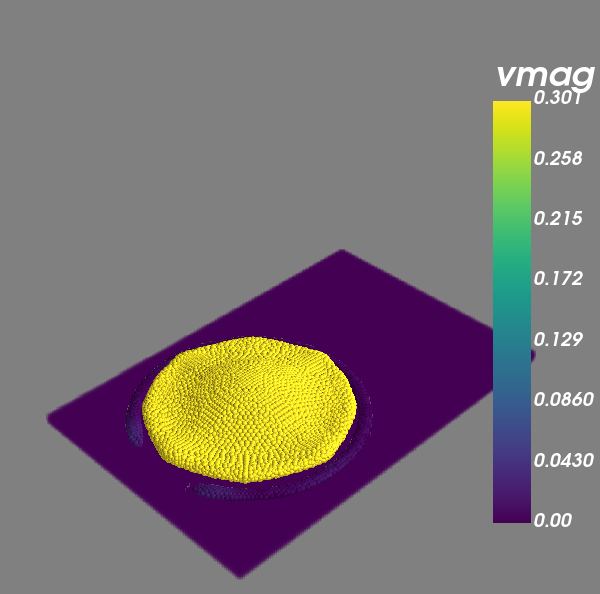}
  \includegraphics[width=0.15\textwidth]{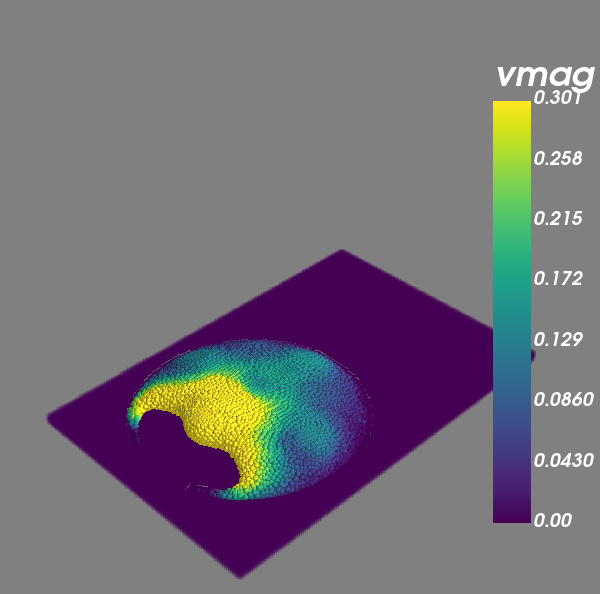}
  \includegraphics[width=0.15\textwidth]{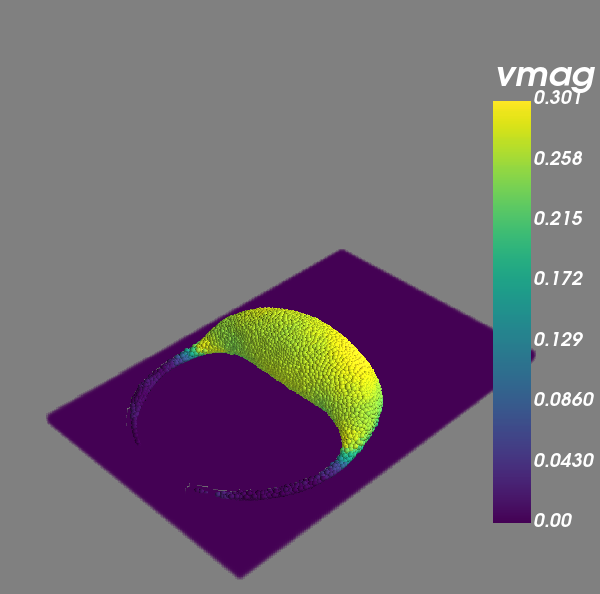}
  \includegraphics[width=0.15\textwidth]{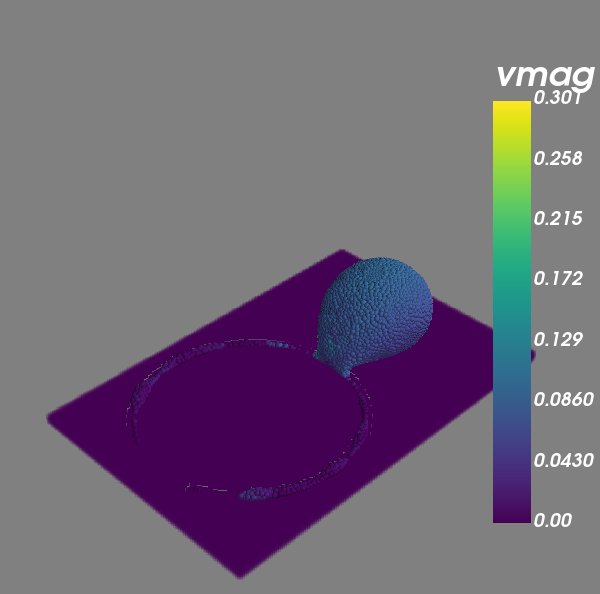}

  \vspace{0.4em}

  \includegraphics[width=0.15\textwidth]{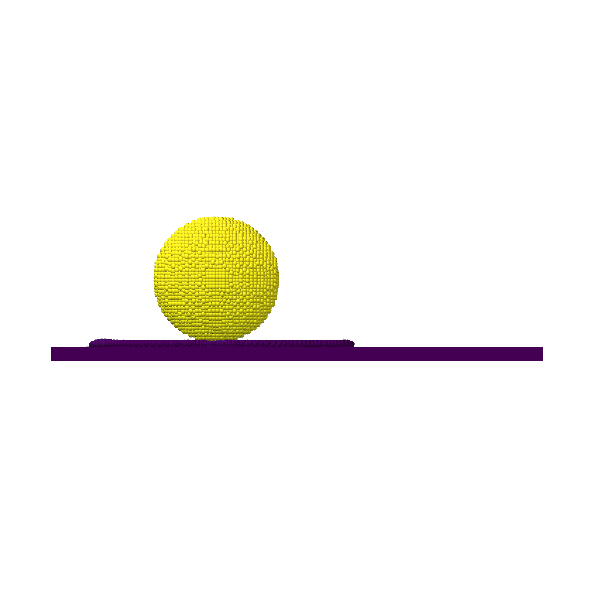}
  \includegraphics[width=0.15\textwidth]{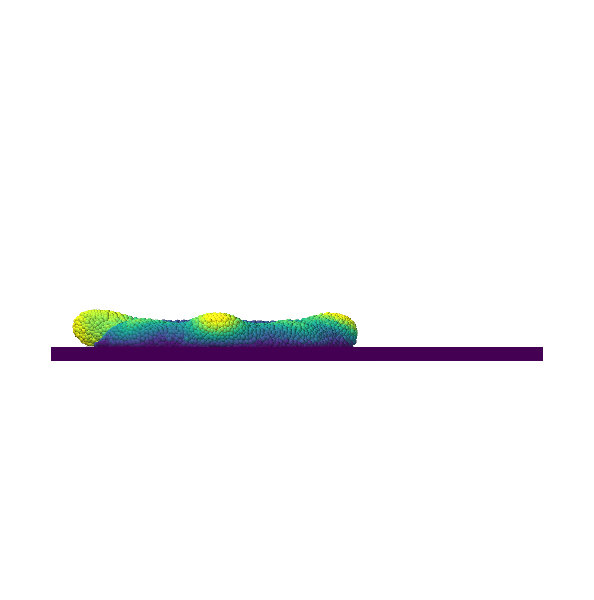}
  \includegraphics[width=0.15\textwidth]{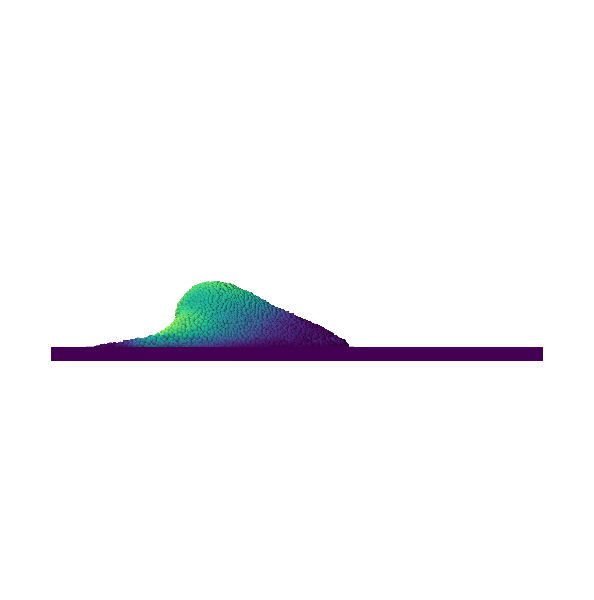}
  \includegraphics[width=0.15\textwidth]{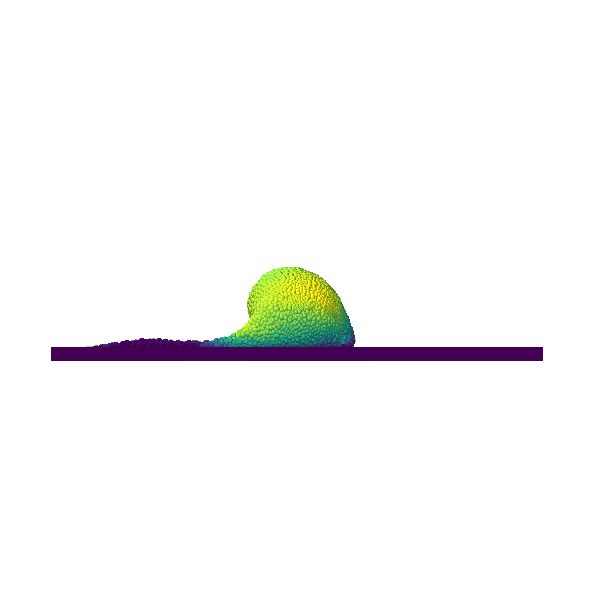}
  \includegraphics[width=0.15\textwidth]{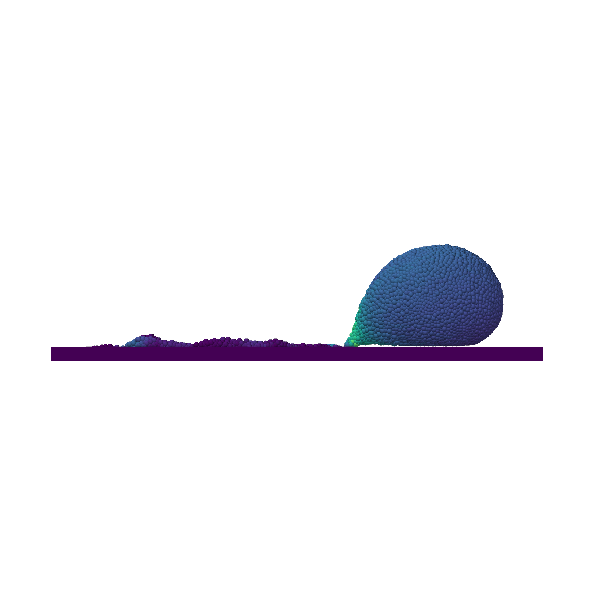}

  \caption{Comparison between experimental observation and numerical simulation of droplet impact and rolling rebound on a patterned wettability surface.}
  \label{fig:exp_num_comparison}
\end{figure} 

Figure~\ref{fig:exp_num_comparison} compares the experimentally observed impact sequence with the corresponding numerical results. The first and second rows display the experimental snapshots, while the third and fourth rows show the simulated droplet evolution at the same representative instants. The second and fourth rows provide the corresponding side views along the \(y\)-direction. The numerical model successfully captures the main stages of the process, including rapid spreading, contact-line pinning along the SHL arc, asymmetric retraction, and the subsequent rolling rebound. In particular, the predicted droplet shapes, interfacial deformations, and overall rebound dynamics agree closely with the experimental observations.

These results demonstrate that the proposed numerical approach can accurately capture the coupled effects of surface-tension-driven flow, wettability heterogeneity, and dynamic contact-line motion on patterned substrates. The good agreement between experiment and simulation further confirms the robustness of the present method for resolving highly transient droplet impact and rebound phenomena on chemically heterogeneous surfaces.

\subsection{Numerical Verification of Coalescence-Induced Droplet Jumping}

To validate the capability of the present model in reproducing transient capillary-driven dynamics, we simulate the coalescence of two unequal droplets on a flat substrate \cite{boreyko2009self}. In this case, a mobile droplet with a diameter of \(270\,\mu\mathrm{m}\) approaches a stationary droplet with a diameter of \(200\,\mu\mathrm{m}\). Upon contact, the two droplets coalesce into a single larger droplet. Although the coalescence is initiated primarily along the substrate plane, the merged droplet subsequently accelerates in the out-of-plane direction and lifts off from the surface with a jumping velocity of approximately \(0.14\,\mathrm{m/s}\).

\begin{figure}[!htbp]
  \centering

  \includegraphics[width=0.7\textwidth]{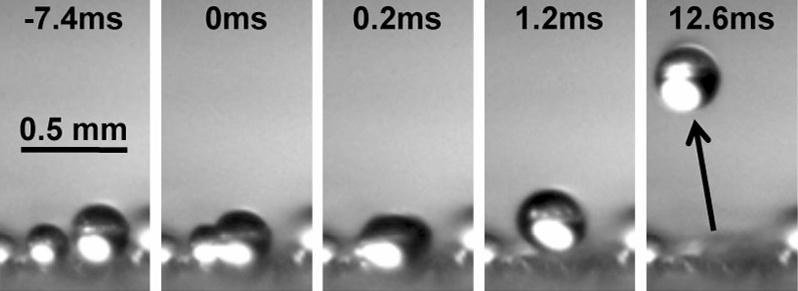} 

  \begin{minipage}{0.78\textwidth}
    \centering
    \includegraphics[width=0.155\textwidth]{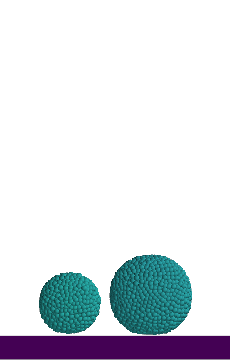}
    \includegraphics[width=0.155\textwidth]{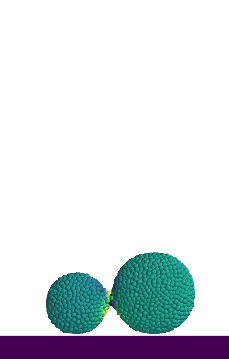}
    \includegraphics[width=0.155\textwidth]{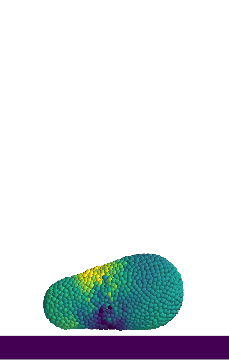}
    \includegraphics[width=0.155\textwidth]{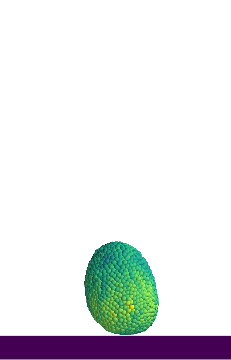}
    \includegraphics[width=0.155\textwidth]{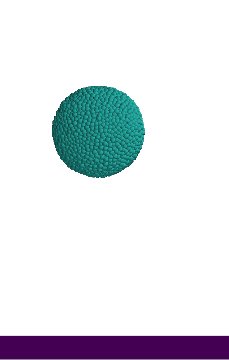}
    \includegraphics[width=0.09\textwidth]{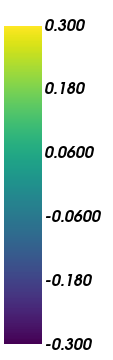}
  \end{minipage}

  \caption{Simulation of coalescence-induced jumping of two droplets on a flat substrate.}
  \label{fig:coalescence_jump}
\end{figure}

Figure~\ref{fig:coalescence_jump} presents the simulated coalescence sequence, where the color contours denote the vertical velocity component \(v_z\). The simulation captures the rapid formation and expansion of the liquid bridge immediately after contact, followed by the retraction of the merged droplet into a more compact shape. During this process, the total interfacial area decreases, leading to the release of excess surface energy. This released energy is partially converted into kinetic energy, which generates a net upward momentum and ultimately causes the droplet to detach from the substrate.

The predicted evolution reproduces the key physical features of coalescence induced jumping, including bridge growth, capillary retraction, and vertical ejection. The distribution of \(v_z\) further confirms that the dominant post-coalescence motion is oriented normal to the surface. These results demonstrate that the present method is capable of accurately resolving fast interfacial deformation and energy conversion during droplet coalescence, thereby providing reliable predictions for dynamic wetting phenomena on solid surfaces.

\section{Conclusion}\label{sec6} 

In this work, a physically grounded SPH model has been developed for the simulation of dynamic droplet behaviors. The proposed method establishes an explicit relationship between the intermolecular potential energy and the surface tension coefficient. This relation enables the microscopic-force modeling of interfacial tension at the liquid--gas interface and adhesion at the liquid--solid interface. In addition, the pressure of solid particles at the solid--liquid interface is evaluated from neighboring fluid particles using a kernel-weighted interpolation with a hydrostatic correction. Together with a single-phase modeling strategy, these features provide an efficient and physically consistent approach for droplet dynamics simulation.

The proposed framework has been assessed through representative droplet wetting and impact problems. The results show that it can accurately reproduce the principal features of droplet dynamics, including spreading, recoiling, and interfacial deformation, with satisfactory agreement with reference solutions. The simulations also demonstrate the good robustness and efficiency of the method in handling interfacial interactions and solid-boundary effects.

Overall, the present work provides an accurate, robust, and efficient SPH approach for dynamic droplet simulations on complex surfaces. It also offers a useful basis for future studies of more complex wetting, impact, and multiphase interfacial flow problems. 

\bibliographystyle{siamplain}
\bibliography{references} 
\end{document}